\documentclass{svproc}
\usepackage{xcolor}
\usepackage{amsmath,amssymb}
\usepackage{epsfig}
\usepackage{graphicx,float}
\usepackage{subcaption}
\usepackage{url}
\usepackage[linesnumbered, ruled, vlined]{algorithm2e}
\usepackage{adjustbox}
\usepackage{float}
\usepackage{tikz}
\usetikzlibrary{arrows.meta}
\usepackage{caption} 
\DeclareMathOperator{\ILMT}{ILMT}

\date{}
\begin{document}

\mainmatter
\authorrunning{Bonato, Carr, Chaudhary, Marbach, and Mishura}
\title{The Iterated Local Model for Tournaments}

\author{Anthony Bonato\inst{1}\thanks{Supported by an NSERC Discovery Grant.}\and MacKenzie Carr\inst{1} \and Ketan Chaudhary\inst{1}
\and  Trent Marbach\inst{2} \and Teddy Mishura\inst{1}}

\institute{Department of Mathematics, Toronto Metropolitan University \\ Toronto, Ontario, Canada
\and
Department of Mathematics and Statistics, Acadia University \\ Wolfville, Nova Scotia, Canada}

\maketitle

\begin{abstract}
Transitivity is a central, generative principle in social and other complex networks, capturing the tendency for two nodes with a common neighbor to form a direct connection. We propose a new model for highly dense, complex networks based on transitivity, called the Iterated Local Model Tournament (ILMT). In ILMT, we iteratively apply transitivity to form new tournaments by cloning nodes and their adjacencies, and either preserving or reversing the orientation of existing arcs between clones. The resulting model generates tournaments with small diameters and high connectivity as observed in real-world complex networks. We analyze subtournaments or motifs in the ILMT model and their universality properties. For many parameter choices, the model generates sequences of quasirandom tournaments. We also study the graph-theoretic properties of ILMT tournaments, including their cop number, domination number, and chromatic number. We finish with a set of open problems and variants of the ILMT model for oriented graphs.
\end{abstract}

\section{Introduction}

Complex networks are pervasive in the real world, and key observed properties include the small-world effect, scale-free degree distributions, and clustering into dense communities. Such networks are often globally sparse, meaning they have few edges relative to their number of nodes. See the book \cite{bbook} for a survey of complex networks. 

In the present work, we consider highly dense networks in which each pair of nodes is adjacent. \emph{Tournaments} are oriented graphs where each pair of nodes shares exactly one directed edge (or arc). Tournaments are simplified representations of highly interconnected structures in networks. For example, we may envision an active sub-thread on the discussion site Reddit, in which users comment on one another's posts. A tournament arises from such social networks by assigning an arc $(u,v)$ if the user $u$ responds more frequently to the posts of $v$ than $v$ does to $u$. Variants include arcs determined by up- or down-votes. Similarly, densely structured communities also occur on other social media platforms, such as X, Instagram, and TikTok, where subsets of users cluster around a given topic or account. Other examples of tournaments in real-world networks include those in sports, where edges represent one player or team winning over another. Such tournaments on social and other networks grow organically; therefore, it is natural to consider models that simulate their evolution. 

\emph{Balance theory} in undirected graphs cites mechanisms to complete triads (that is, subgraphs consisting of three nodes) in social and other complex networks \cite{easl,he}. A central mechanism in balance theory is \emph{transitivity}: if $x$ is a friend of $y,$ and $y$ is a friend of $z,$ then $x$ is a friend of $z$; see, for example, \cite{scott}. Directed networks of ratings or trust scores, along with models for their propagation, were first considered in \cite{guha}. \emph{Status theory} for directed networks, first introduced in \cite{lesk}, was motivated by both trust propagation and balance theory. While balance theory focuses on likes and dislikes, status theory posits that an arc indicates that the creator of the link views the recipient as having higher status. For example, on X or other social media platforms, an arc denotes one user following another, and the follower may have higher social status. Evidence for status theory was found in directed networks derived from Epinions, Slashdot, and Wikipedia \cite{lesk}. 

Many models simulate the properties of complex networks, such as preferential attachment \cite{ba,bol}. The \emph{Iterated Local Transitivity} (\emph{ILT}) model introduced in \cite{ILT} and further studied in \cite{ILM,ILAT} simulates structural properties in complex networks that emerge from transitivity. Transitivity gives rise to the notion of \emph{cloning}, in which a new node $x$ is adjacent to all the neighbors of an existing node $y$. Note that in the ILT model, nodes exert local influence within their neighborhoods. The ILT model simulates many properties of social networks. For example, as shown in \cite{ILT}, the model-generated graphs densify over time and exhibit poor spectral expansion. In addition, the ILT model generates graphs with the small-world property, which exhibit low diameter and high clustering coefficient relative to random graphs with the same number of nodes and the same expected average degree. A directed analogue of the ILT model was introduced in \cite{ILDT}, and a tournament-based model was first considered in \cite{ILTT}. Other variants of the ILT model include \cite{ILH,FM,IGM,itind}.

We introduce a new deterministic model for complex networks based on transitivity that yields tournaments, thereby extending the model considered in \cite{ILTT}. The \emph{Iterated Local Model for Tournaments} (or \emph{ILMT}) model is defined as follows. Fix a \emph{base} tournament $G_0$ and an infinite binary sequence $s$ named the \emph{generating sequence}. We define a sequence of tournaments $G_t$ (referred to as \emph{time-steps}) as follows. 

For each time-step $t\geq 1$, we obtain the tournament $G_t$ from $G_{t-1}$ by the following steps.
\begin{enumerate}
\item Add a node $x'$ for each node $x$ of $G_{t-1}$.
\item Add the arcs $(u,v'),(u',v)$ for each arc $(u,v)$ in $G_{t-1}$.
\item Add the arc $(x',x)$ for each node $x$ in $G_{t-1}$.
\item For each arc $(u,v)$ in $G_{t-1}$, add the arc $(u',v')$ if $s(t)=1$ or add the arc $(v',u')$ if $s(t)=0$. 
\end{enumerate}
We refer to $x'$ as the \textit{clone} of $x$ in $G_{t}$ and $x$ as the \textit{parent} of $x'$. The tournaments $G_t$ are called \emph{ILMT tournaments}. See Figure~\ref{figILMT} for examples of ILMT tournaments.
\begin{figure}[ht]
    \centering
   \scalebox{0.7}{\begin{tikzpicture}[
    vertex/.style={circle, fill=black, inner sep=1.5pt},
    edge/.style={->, >={latex[length=2.5mm, width=1.5mm]}, shorten >=1pt, shorten <=1pt}
    ]
    
    \begin{scope}
    \node[vertex, label=below:$a$] (a) at (0,0) {};
    \node[vertex, label=above:$b$] (b) at (0,2) {};
    \path[edge] (a) edge (b);
    \node at (0, -1.25) {$G_0$};
\end{scope}

\begin{scope}[xshift=2.5cm]
    
    \node[vertex, label=below:$a$] (a) at (0,0) {};
    \node[vertex, label=above:$b$] (b) at (0,2) {};
    \node[vertex, label=below:$a'$] (ap) at (2,0) {};
    \node[vertex, label=above:$b'$] (bp) at (2,2) {};

    \path[edge] (a) edge (b);   
    \path[edge] (a) edge (bp);  
    \path[edge] (ap) edge (b);  
    \path[edge] (ap) edge (a);  
    \path[edge] (bp) edge (b);  
    \path[edge] (ap) edge (bp); 
    
    \node at (1, -1.25) {$G_1$};
\end{scope}
\begin{scope}[xshift = 7cm]
    
    \node[vertex, label=below:$a$] (a) at (0, 0) {};
    \node[vertex, label=above:$b$] (b) at (0, 2) {};
    \node[vertex, label=below:$a'$] (ap) at (2, 0) {};
    \node[vertex, label=above:$b'$] (bp) at (2, 2) {};
    \node[vertex, label=below:$a''$] (app) at (4, 0) {};
    \node[vertex, label=above:$b''$] (bpp) at (4, 2) {};
    \node[vertex, label=below:$(a')'$] (apprime) at (6, 0) {};
    \node[vertex, label=above:$(b')'$] (bpprime) at (6, 2) {};


    \path[edge] (a) edge (b);
    \path[edge] (a) edge (bp);
    \path[edge] (ap) edge (b);
    \path[edge] (ap) edge (a);
    \path[edge] (bp) edge (b);
    \path[edge] (ap) edge (bp);

    \path[edge] (bpp) edge (app);
    \path[edge] (bpprime) edge (app);
    \path[edge] (bpp) edge (apprime);
    \path[edge] (app) edge (apprime);
    \path[edge] (bpp) edge (bpprime);
    \path[edge] (bpprime) edge (apprime);
    
    \path[edge] (app) edge [bend left=20] (a);
    \path[edge] (bpp) edge [bend right=20] (b);
    \path[edge] (apprime) edge [bend left=20] (ap);
    \path[edge] (bpprime) edge [bend right=20] (bp);
    
    \path[edge] (a) edge (bpp);
    \path[edge] (app) edge (b);
    \path[edge] (a) edge (bpprime);
    \path[edge] (app) edge (bp);
    \path[edge] (ap) edge (bpp);
    \path[edge] (apprime) edge (b);
    \path[edge] (ap) edge (app);
    \path[edge] (apprime) edge [bend left] (a);
    \path[edge] (bp) edge (bpp);
    \path[edge] (bpprime) edge [bend right] (b);
    \path[edge] (ap) edge (bpprime);
    \path[edge] (apprime) edge (bp);
    
    \node at (3, -1.25) {$G_2$};
    \end{scope}
\end{tikzpicture}}
    \caption{The first three time-steps of ILMT tournaments, where $G_0$ is a directed edge, and the sequence $s$ has values $s(1) = 1$ and $s(2) = 0.$ In $G_2$, $x''$ denotes the clone of $x$ while $(x')'$ denotes the clone of $x'$, for $x=a,b.$ }
    \label{figILMT}
\end{figure}
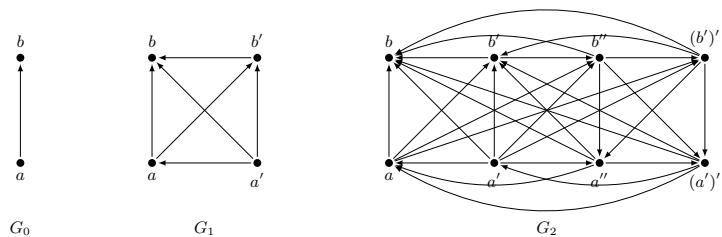

The \textit{reverse} of $G_{t-1}$ if $s(t)=0$. Otherwise, the subtournament of $G_t$ induced by the clones is isomorphic to $G_{t-1}$. We denote $G_t$ more precisely by $\ILMT_{t,s}(G_0)$. For brevity, we often omit explicit references to $G_0$ and the generating sequence $s$. We refer to the entries of $s$ with a zero as its \emph{support}; if there are infinitely many 0-entries, we say that $s$ has \emph{infinite support}. 

The paper is organized as follows. In Section~2, we prove that ILMT tournaments have low diameter and possess high connectivity. We show that ILMT tournaments are universal in the sense that they embed all finite tournaments. We give precise counts in ILMT tournaments of the number of directed 3-cycles and of linear orders of order 3. Quasirandomness is explored in Section~3.  If the generating sequence has infinite support, then ILMT tournaments are quasirandom; see Theorem~\ref{thm:0stepquasi}. In particular, this gives precise, asymptotic counts of any given finite tournament in ILMT tournaments. In Section~4, we consider graph-theoretical properties of ILMT tournaments, including their cop numbers, domination numbers, and chromatic numbers. Our final section includes directions for future research.

Throughout the paper, we consider finite, simple tournaments. A tournament is \emph{strong} if for each pair of nodes $x$ and $y$, there are directed paths connecting $x$ to $y$ and $y$ to $x$. The \emph{in-neighbors} and \emph{out-neighbors} of node $v$ are denoted by $N^-(v)$ and $N^+(v),$ respectively. If $x_1,x_2,\dots ,x_n$ are nodes of a tournament $G$ so that $(x_i,x_j)$ is an arc whenever $i<j$, then $G$ is a \emph{linear order}. We denote the linear order on three nodes by $T_3$ and the directed 3-cycle by $D_3.$ A \emph{sink} is a node with out-degree 0, and a \emph{source} is a node with in-degree 0. For background on tournaments, the reader is directed to \cite{bang-jensen-gutin-2009} and to \cite{west} for background on graphs.  For background on social and complex networks, see \cite{bbook,at}. 

\section{Diameter, Connectedness, and Motifs}

In this section, we show that ILMT tournaments satisfy several properties observed in real-world complex networks, including small diameter and high connectivity. We first show that, under mild restrictions on the base tournament and generating sequence, ILMT tournaments have diameter at most 3.

\begin{theorem}\label{ILMTdiam}
If the tournament $G_0$ has no sink and $s$ has nonempty support, then for sufficiently large $t$, the diameter of $\mathrm{ILMT}_{t,s}(G_0)$ is at most 3.
\end{theorem}
\begin{proof}
Let $G_0$ be a base tournament with no sink and let $G_t = \mathrm{ILMT}_{t,s}(G_0)$, for each $t\geq 1$. In particular, $G_0$ must have at least three nodes and $G_t$ has no sink for every $t\geq 0$, since a clone at time $t$ has its parent as an out-neighbor.
    
Suppose $s(t+1) = 0$ for some $t \geq 0$. We show that $G_{t+1}$ must be a strong tournament of diameter at most 3. Pick distinct nodes $u$ and $v$ in $G_{t}$ such that $(u,v)$ is an arc. Since there are no sinks, there must be a third node $w$ in $G_{t}$ that is an out-neighbor of $v$. The nodes $u,v',$ and $u'$ form a $D_3$ in $G_{t+1}$. We also have that either the nodes $u,v,w'$ form a $D_3$ in $G_{t + 1}$ or they form a directed 4-cycle along with $u'$, depending on whether $w$ is an in-neighbor or out-neighbor of $u$, respectively. Altogether, we have that any two nodes in $G_{t+1}$ are in a directed cycle of length at most 4. This is equivalent to the condition that $G_{t+1}$ is strong with diameter at most 3.

Next, we show that if $G_t$ is a strong tournament of diameter at most 3 and $s(t+1) = 1$, then $G_{t+1}$ has diameter at most 3. Let $u,v\in V(G_t)$, $\alpha\in \{u,u'\}$, and $\beta\in \{v,v'\}$, then a directed path from $\alpha $ to $\beta$ can be obtained from a directed path in $G_t$ of length at most $3$ from $u$ to $v$ by replacing $u$ with $\alpha$ and replacing $v$ with $\beta$. The length of the resulting directed path in $G_{t + 1}$ is also at most $3$. Therefore, $G_{t + 1}$ is a strong tournament of diameter at most 3. 

Proceeding inductively, we have that if $t_0 \geq 0$ such that $s(t_0 + 1) = 0$, then $G_t$ is a strong tournament of diameter at most $3$ for all $t > t_0$. Since the support of $s$ is nonempty, such a $t_0$ must exist.\qed
\end{proof}

The \emph{connectivity} $\kappa(T)$ of $T$ is the minimum size of a set of nodes $S \subseteq V(T)$ such that $T \setminus S$ is not strongly connected. A tournament $T$ is called \emph{$k$-connected} if $\kappa(T) \ge k$. As a consequence of the following theorem, the connectivity of ILMT is unbounded as time-steps become large. The proof is omitted for space considerations.

\begin{theorem}
If $G_0$ is a $k$-connected tournament, for some $k\geq 1$, and $G_1 = \mathrm{ILMT}_{1,s}(G_0)$ for some generating sequence $s$, then $G_1$ is $2k$-connected.
\end{theorem}

A \emph{motif} is an isomorphism type of a specified tournament. Motifs play an important role in the structure of complex networks, both in graphs and in digraphs; see \cite{Milo}. We provide precise counts of motifs of order 3.

\begin{theorem}\label{ILMTtriples}
Let $G_0$ be a tournament of order $n_0\geq3$, let $s$ be a generating sequence, and for each $t\geq 1$, let $G_t = \mathrm{ILMT}_{t,s}(G_0)$. Let $a_t$ denote the number of subsets of $V(G_t)$ of cardinality 3 that form a $D_3$ and let $b_t$ denote the number of subsets of $V(G_t)$ of cardinality 3 that form a $T_3$. We then have that for all $t\geq 1$,
$$a_t = 2^{s(t)+2}a_{t-1}+(1-s(t)){2^{t-1}n_0+1\choose 3},$$
and for all $t\geq 0$,
$$b_t = {2^tn_0\choose 3} - a_t.$$

\end{theorem}

\begin{proof}
Let $n_t=2^tn_0$ denote the order of $G_t$, for each $t\geq 0$. Since each $G_t$ is a tournament, $a_t + b_t$ is equal to the number of subsets of $V(G_t)$ of cardinality 3. Therefore, $b_t = {n_{t}\choose 3} - a_t$, for $t\geq 0$. It suffices to show that for $t\geq 1$,
    $$a_t = \begin{cases}
       8a_{t-1} & \text{if }s(t) = 1,\\
       5a_{t-1} + b_{t-1} + {2^{t-1}n_0\choose 2} & \text{if }s(t) = 0.
    \end{cases}$$

Suppose that $s(t) = 1$. No $D_3$ in $G_t$ can consist of an arc of the form $(x',x)$ for $x\in V(G_{t-1})$. Therefore, every subset of $V(G_t)$ of cardinality 3 that forms a $D_3$ in $G_t$ is one of the eight sets obtained from the nodes of a $D_3$ in $G_{t-1}$ by replacing any number of them with their clones. Therefore, $a_t = 8a_{t-1}$. 

Next, suppose that $s(t) = 0$. Fix a subset $S$ of $V(G_{t-1})$ of cardinality 3 that forms a $D_3$. Of the eight sets that one can construct from $S$ by replacing any number of its elements with its clone in $G_t$, only five form a $D_3$. If exactly two nodes of $S$ are replaced with their clones, then the subtournament of $G_t$ induced by the resulting set of nodes is a $T_3$. 
    
If $S$ forms a $T_3$ in $G_{t-1}$, then of the eight sets that one can construct from $S$ by replacing any number of its elements with their clones in $G_t$, we have that seven form a $T_3$ in $G_t$. The only one of these eight sets that forms a $D_3$ in $G_t$ is the one obtained by replacing exactly the source and the sink of $T$ with their clones. Allowing $S$ to range over all cardinality 3 subsets of $V(G_{t-1})$ that form a $D_3$ in $G_{t-1}$ we have that there are $5a_{t-1} + b_{t-1}$ $D_3$'s in $G_t$ that do not contain an arc from a clone to its parent.

Finally, still under the case $s(t) = 0$, suppose that $S$ is a 3-node subtournament of $G_t$ containing some $x\in V(G_{t-1})$ and its clone $x'$. We then have that $S$ forms a $T_3$ if and only if the third node of $T$ is either a neighbor of $x$ in $G_{t-1}$ or the clone of an in-neighbor of $x$ in $G_{t-1}$. If $u$ is an out-neighbor of $x$ in $G_{t-1}$, then $\{x,x',u'\}$ forms a $D_3$. Therefore, the number of cardinality 3 subsets of $V(G_t)$ that form a $D_3$ and contain both a clone and its parent is equal to the total out-degree of $G_{t-1}$, which is equal to the size of $G_{t-1}$. Altogether, we have that 
$$a_t = 5a_{t-1} + b_{t-1} + {2^{t-1}n_0\choose 2},$$
when $s(t) = 0$. Substituting $b_{t-1} = {2^{t-1}n_0 \choose 3} - a_{t-1}$, we obtain the result. \qed
\end{proof}

For a $k$-node subtournament $H$ of $G,$ its \emph{proportion} is the ratio of the number of copies of $H$ in $G$, and the number of $k$-element subsets of $V(G).$ From Theorem~\ref{ILMTtriples}, we derive the asymptotic proportion of motifs of order 3 in ILMT tournaments. The following theorem does so explicitly for the proportion $D_3$'s, with the proportion of $T_3$'s derived by taking the difference with 1. Note that, as expected, we observe that ILMT typically generates a large proportion of $T_3$'s.

\begin{corollary}\label{prop3}
Let $G_0$ be a tournament of order $n_0$ and let $\mu$ be the proportion of $D_3$'s in $G_0$. 
\begin{enumerate}
\item If $s(t)=1$ for all $t\geq 1$, then the proportion of $D_3$'s in $\mathrm{ILMT}_{t,s}(G_0)$ approaches $\frac{n_0(n_0-1)(n_0-2)}{n_0^3}\mu$ as $t\rightarrow \infty$.
\item If $s$ has infinite support, then the proportion of $D_3$'s in $\mathrm{ILMT}_{t,s}(G_0)$ approaches $\frac{1}{4}$ as $t\rightarrow \infty$.
\end{enumerate}  
\end{corollary}

\begin{proof}
By Theorem~\ref{ILMTtriples} we have that $$\frac{a_t}{8^t} = 2^{s(t)+2}\frac{1}{8}\cdot\frac{a_{t-1}}{8^{t-1}}+\frac{1-s(t)}{8}\frac{1}{2^{3(t-1)}}{2^{t-1}n_0+1\choose 3}.$$ Note that $\frac{1}{2^{3(t-1)}}{2^{t-1}n_0+1\choose 3}\rightarrow \frac{n_0^3}{6}$ as $t\rightarrow \infty$. Let $\hat{a}_t = \frac{a_t}{8^t}$. We then have that $$\hat{a}_t = O(2^{-t})+\begin{cases}
    \hat{a}_{t-1} & \text{if }s(t) = 1,\\
    \frac{\hat{a}_{t-1}}{2} + \frac{n_0^3}{6\cdot 8} & \text{if }s(t) = 0.
\end{cases}$$ 

If $s$ is the constant sequence of 1's, then $\hat{a}_t$ remains constantly $a_0$, so that the proportion of $D_3$'s $\frac{a_t}{{{2^tn_0}\choose3}}$ approaches $$\frac{6a_0}{n_0^3} = \frac{n_0(n_0-1)(n_0-2)a_0}{n_0^3 {n_0 \choose 3}} = \frac{n_0(n_0-1)(n_0-2)}{n_0^3}\mu$$ as $t\rightarrow \infty$. If $s$ has infinite support, then $\hat{a}_t$ approaches the fixed point of the function $x \mapsto \frac{x}{2} + \frac{n_0^3}{48}$, which is $\frac{n_0^3}{24}$. In this case, the proportion of $D_3$'s in $G_t$ approaches $\frac{n_0^3}{24}\cdot\frac{6}{n_0^3} = \frac{1}{4}$. \qed
\end{proof}

If we perform enough $0$-steps, then we arrive at the following \emph{universality} property: any motif type can appear eventually in ILMT tournaments after sufficiently many time-steps. 

\begin{theorem}
Let $G_0$ be a tournament of order $n_0$ and let $n\leq n_0$. For a generating sequence $s$ with infinite support, let $r$ be the smallest integer such that $$\left|\{1 \leq t \leq r : s(t)=0\}\right| = n.$$ We then have that every tournament on $n$ nodes is isomorphic to a subtournament of $\mathrm{ILMT}_{r,s}(G_0)$.
\label{lemma:0steps}
\end{theorem}

\begin{proof}
Let $S_0$ be an $n$-subset of the nodes in $G_0$, and let $G$ be a tournament with $n$ nodes. Let $\phi_0:S_0\to V(G)$ be a bijection. Choose some node $u$ of $G$. Let $u_0 = \phi_0^{-1}(u)$ and $t_1 = \min\{t : s(t)=0\}$. Define $A_0$ to be the set of nodes $v\in S_0$ such that $(u_0,v)$ is an arc in $G_0$, and $(\phi_0(v),u)$ is an arc in $G$. 

Now, define $S_1$ to be the $n$-subset of the nodes of $G_{t_1}$ formed from $S_0$ by replacing the elements of $A_0\cup\{u_0\}$ by their clones in $G_{t_1}$. Let $\phi_1:S_1 \to V(G)$ be the bijection given by $\phi_1(v) = \phi_0(v)$ for $v\in S_0\cap S_1$, and $\phi_1(v') = \phi_0(v)$ for $v\in S_0 \setminus S_1$, where $v'$ is the clone of $v$ in $G_{t_1}$. Now, for every $v\in S_1$, we have that $(\phi_1(u_0'),\phi_1(v)) = (u,\phi_1(v))$ is an arc in $G$. 

We repeat this process of replacing arcs by their reverse, each time choosing a node in $G$ that has not been chosen previously and defining the set $S_i$ and the bijection $\phi_i :S_i \to V(G)$ analogously. After $n$-many 0-steps, we have that for each arc $(x,y)$ in $\mathrm{ILMT}_{r,s}(G_0)$ with $x,y\in S_n$, $(\phi_n(x),\phi_n(y))$ is an arc in $G$. \qed \end{proof}

We have the following corollary.

\begin{corollary}\label{couni}
    If a generating sequence $s$ has infinite support, then every tournament on $n$ nodes is a subtournament of some $\mathrm{ILMT}$ tournament.
\end{corollary}

\section{Quasirandomness}

Given Corollary~\ref{prop3} and \ref{couni}, an interesting question is to estimate the counts of larger motifs. Although this problem appears challenging, we provide an asymptotically exact count for a broad range of ILMT tournaments. To this end, we turn to the theory of quasirandom tournaments, first studied in \cite{fan}. Such tournaments are deterministic families that share several asymptotic properties with random tournaments. 

Given tournaments $G$ and $H$, we define the \emph{density} of $H$ in $G$ as the following ratio, where $n(H,G)$ is the number of copies of $H$ present in $G$, and $\mathrm{Aut}(H)$ is the automorphism group of $H$:
$$
d_G^*(H) = \frac{|\mathrm{Aut}(H)|n(H,G)}{|V(H)|!\binom{|V(G)|}{|V(H)|}}.
$$

A sequence of tournaments $(G_n :n\geq 1)$ is said to be a \emph{quasirandom} if for all tournaments $H$ we have that 
$$
\lim_{n\to\infty}d_{G_n}^*(H) = 2^{-\binom{|V(H)|}{2}}.
$$
The uniformly random tournament on $n$ nodes labeled $1,2,\ldots ,n$ has arcs $(i,j)$ with $i<j$ chosen with probability $p=1/2$. In a quasirandom sequence, the density of any subtournament $H$ in $G_n$ approaches the density of $H$ in the uniformly random tournament. While calculating the density of every subtournament may appear challenging, it may be unnecessary. A tournament $G$ is \emph{quasirandom-forcing} if any sequence of tournaments $(G_n :n\geq 1)$ such that 
$$
\lim_{n\to\infty}d_{G_n}^*(G) = 2^{-\binom{|V(G)|}{2}},
$$
is quasirandom. Quasirandom-forcing tournaments have been completely characterized in the following result. 

\begin{theorem}[\cite{raz}]\label{thm:quasiforce}
If $G$ is a quasirandom-forcing tournament, then $G$ is either a linear order with at least four nodes or a particular non-transitive tournament with five nodes.
\end{theorem}

Hence, to show that a sequence of tournaments is quasirandom, it is sufficient to consider the density of the linear order on four nodes. We take that approach in the following theorem.

\begin{theorem}\label{thm:0stepquasi}
If $s$ has infinite support, then $(G_t :t \geq 0)$ forms a quasirandom sequence.
\end{theorem}

\begin{proof}
For simplicity, we consider the sequence consisting of all 0's only. Now, there are four unique tournaments on four nodes: the linear order, the \emph{Winner} tournament, which is the unique non-transitive tournament with a node of out-degree 3, the \emph{Loser} tournament, which is the unique non-transitive tournament with a node of in-degree 3, and the \emph{Mixed} tournament, which is the unique tournament with out-degree sequence $(1,1,2,2)$. 
    
Let $G$ be any tournament, and let $G_t$ be the tournament obtained from $t$ applications of 0-steps to $G$. Let $p_t$ be the proportion of 4-tuples in $G_k$ that are linear orders, let $w_t$ be the proportion of 4-tuples in $G$ that are Winner, let $\ell_t$ be the proportion of 4-tuples in $G_t$ that are Loser, and let $m_t$ be the proportion of 4-tuples in $G_t$ that are Mixed. 

Consider any 4-tuple $X=\{x_1,x_2,x_3,x_4\}$ of nodes in $G_{t}.$ Each $x_i$ is either a node in $G_{t-1}$ or a clone of a node in $G_{t-1}$, and so we can create a \textit{corresponding tournament} $X^*$ that is contained entirely in $G_{t-1}$ by replacing each clone node with its parent node in $G_{t-1}.$ The tournament type of $X$ implies the tournament type of $X^*$ --- for example, if $X=\{x_1',x_2',x_3',x_4'\}$ forms a Winner tournament in $G_t$, then $X^*=\{x_1,x_2,x_3,x_4\}$ forms a Loser tournament in $G_{t-1}.$ Thus, for each of the 16 types of 4-node tournaments in $G_t$, we determine the type of its corresponding tournament in $G_{t-1}$. Each proportion of the different tournament types in $G_t$ can therefore be expressed as a linear combination of the proportion of tournament types in $G_{t-1}$, which we state as the following matrix equation:
    
\begin{equation*}
    \sigma_t =\begin{bmatrix}
            p_t \\
            w_t \\
            \ell_t \\
            m_t
        \end{bmatrix}=
        \begin{bmatrix}
            \frac{11}{16} & \frac{1}{16} & \frac{1}{16} & \frac{3}{16} \\
            \frac{3}{16} & \frac{9}{16} & \frac{1}{16} & \frac{3}{16} \\
            \frac{3}{16} & \frac{1}{16} & \frac{9}{16} & \frac{3}{16} \\
            \frac{3}{16} & \frac{1}{16} & \frac{1}{16} & \frac{11}{16} 
        \end{bmatrix}
        \begin{bmatrix}
            p_{t-1} \\
            w_{t-1} \\
            \ell_{t-1} \\
            m_{t-1}
        \end{bmatrix} = T\sigma_{t-1}
\end{equation*}

We note that $\sigma_t$ only describes the proportion of 4-tournaments $X$ in $G_t$ that do not simultaneously contain a node and its clone. For example, regardless of what type of 4-tournament is formed by the nodes $X = \{x,x',y,z\}$, its corresponding tournament $X^*=\{x,y,z\}$ will not form a 4-tournament. However, the proportion of these tournaments is $\binom{2^t |V(G)|}{3}/\binom{2^t |V(G)|}{4} \sim \frac{4}{2^t|V(G)|},$ which quickly goes to 0 as $t\to\infty$. Therefore, no matter the behavior of these tournaments, they do not meaningfully contribute to the overall proportion of 4-tournament types in $G_t$. Hence, the above matrix equation fully governs the long-term behavior of tournament types in $G_t.$

We note that the matrices $\sigma_t$ form an irreducible aperiodic Markov chain on a finite number of states with transition matrix $T$. It is well known that each such Markov chain has a limiting distribution $\pi$ that is also a \emph{stationary distribution} (see, for example, \cite[Corollary 7.9]{Privault2018}) that is the unique solution to the matrix equation $\pi = T\pi.$ Hence, independent of the initial proportions of each of the types of 4-node tournaments in $G_0$, the proportions $\sigma_t$ must converge to this distribution $\pi$. It is not difficult to show that $\pi^T = [\frac{3}{8}~\frac{1}{8}~\frac{1}{8}~\frac{3}{8}],$ and so as $t$ tends to $\infty$, the proportion 
$$p_t = \frac{n(T_4,G_t)}{\binom{|V(G_t)|}{4}} \to \frac{3}{8}.$$
Hence, we have that
$$ d^*_{G_k}(T_4) = \frac{|\mathrm{Aut}(T_4)|}{|V(T_4)|!}p_k = \frac{1}{24}p_t \to \frac{1}{64} = 2^{-\binom{|V(T_4)|}{2}}.$$
The proof follows from Theorem~\ref{thm:quasiforce}. \qed
\end{proof}

We finish the section by noting that the ILMT process yields, in a certain sense, continuum-many pairwise non-isomorphic sequences of tournaments.  Hence, combined with Theorem~\ref{thm:0stepquasi}, the ILMT model yields continuum-many non-isomorphic sequences of quasirandom tournaments.

\begin{theorem}
Given a tournament $G_0$ and two distinct generating sequences $s$ and $s'$, there is some $t\geq 1$ such that $\mathrm{ILMT}_{t,s}(G_0)$ is not isomorphic to $\mathrm{ILMT}_{t,s'}(G_0)$.
\end{theorem}

\begin{proof} Let $G_t=\mathrm{ILMT}_{t,s}(G_0)$ and let $G_{t_0}=\mathrm{ILMT}_{t,s'}(G_0)$. If $G_1$ is not isomorphic to $G'_1$, then the desired result holds with $t=1$. It suffices to show that if $t_0 > 0$ such that $G_{t_0} = G'_{t_0}$ and if $s(t_0+1) \neq s'(t_0+1)$, then $G_{t_0+1}$ is not isomorphic to $G'_{t_0  + 1}$. 

Assume $t_0 > 0$ such that $G_{t_0} = G'_{t_0}$ and $s(t_0+1) = 1 = s'(t_0+1)+1$. Let $r$ be the number of nodes in $G_{t_0}$  with out-degree $2^{t_0-1}n_0$. A node with out-degree $\delta$ at time $t_0$ must have double the out-degree, $2\delta$, at time $t_0 + 1$. Every clone in $G'_{t_0 + 1}$ has $2^{t_0}n_0 - 1$ out-neighbors other than its parent. Thus, $G_{t_0+1}'$ has $r+2^{t_0}n_0$ nodes with out-degree $2^{t_0}n_0$. However, note that $G_{t_0+1}$ has only $r$ nodes with out-degree $2^{t_0}n_0$, since all clones in a 1-step have odd out-degree. Altogether, we have that the tournaments $G_{t_0+1}$ and $G_{t_0+1}'$ are not isomorphic, since they do not have the same out-degree distribution. That is, the desired result holds for $t=t_0$.\qed
\end{proof}

\section{Graph-theoretical properties}

We next present results on classical graph-theoretical properties of ILMT tournaments. For a tournament $G$, we say that $S\subseteq V(G)$ is \textit{in-dominating} in $G$ if every node in $V(G) \setminus S$ has an in-neighbor in $S$. Out-dominating sets are defined analogously. An in-dominating (respectively, out-dominating) set for $G$ is said to be \textit{minimal} if it does not contain any proper subset that is also in-dominating (respectively, out-dominating) for $G$. The \textit{in-domination number} of $G$, denoted $\gamma^-(G)$, is the least cardinality among in-dominating sets for $G$. The \textit{out-domination number} of $G$, denoted $\gamma^+(G)$, is defined analogously. In the case that $G$ is a tournament, the condition that $\gamma^-(G)=1$ is equivalent to the condition that $G$ has a source. Similarly, when $G$ is a tournament, then $\gamma^+(G) = 1$ if and only if $G$ has a sink.

We provide the following lemma.

\begin{lemma}\label{ILMTindom1}\label{ILMTindom2}
Let $G_0$ be a tournament, let $G_1$ be the tournament obtained from $G_0$ as the result of a 0-step, and for any $S \subseteq V(G_0)$ let $S' = \{x':x\in S\}$. We then have the following.
\begin{enumerate}
    \item $S' = \{x':x\in S\}$ is an in-dominating set for $G_1$ if and only if $S$ is both an in- and out-dominating set for $G_0.$

    \item If $S$ with $|S|>1$ is a minimal in-dominating set for $G_0$, then $S'$ is an in-dominating set for $G_1$.
\end{enumerate}
\end{lemma}

\begin{proof}
    For space considerations, we only prove (1). If $S'$ is an in-dominating set for $G_1$, then for every $x\in V(G_0)\setminus S$, there exist some $u,v \in S$ such that $(u',x)$ and $(v',x')$ are arcs in $G_1$. We have that $(u,x)$ and $(x,v)$ are both arcs in $G_0$, so $S$ is both in- and out-dominating for $G_0$. 

    Conversely, if $S$ is both in- and out-dominating for $G_0$, then every clone in $V(G_1)\setminus S'$ has an in-neighbor in $S'$ that is an out-neighbor of its parent, and every node in $V(G_0)$ has an in-neighbor in $S$ that is the clone of one of its neighbors in $S$. Moreover, every element of $S'$ is an in-neighbor of its parent in $S$. Thus, $S'$ is an in-dominating set for $G_1$. Item (1) follows.
   \qed
\end{proof}



\begin{theorem}\label{ILMTindom3}
 Let $G_0$ be a tournament and $G_1$ the tournament obtained from $G_0$ as the result of a 0-step. For any $S \subseteq V(G_0)$ with $|S|>1$, we have that $S' = \{x':x\in S\}$ is a minimal in-dominating set for $G_1$ if and only if $S$ is a minimal in-dominating set for $G_0$.
\end{theorem}

\begin{proof}
Suppose that $S'$ is a minimal in-dominating set for $G_1$. By Lemma~\ref{ILMTindom1}, $S$ is an in-dominating set for $G_0$. In particular, $S$ contains a minimal in-dominating set for $G_0$, say $S_0$. By Lemma~\ref{ILMTindom2}, the set $S_0' = \{x':x\in S_0\}$ is an in-dominating set for $G_1$ that is contained in $S'$. Since $S'$ is minimal in $G_1$, we have $S'=S_0'$ and thus, $S = S_0$. Therefore, $S$ is a minimal in-dominating set for $G_0$.

Conversely, suppose that $S$ is a minimal in-dominating set for $G_0$. By Lemma~\ref{ILMTindom2}, $S'$ is an in-dominating set for $G_1$. Let $S_0\subseteq S$ such that $S_0' = \{x':x\in S_0\}$ is a minimal in-dominating set for $G_1$. By Lemma \ref{ILMTindom1}, $S_0$ is in-dominating for $G_0$. However, since $S$ is minimal, we have that $S_0 = S$ and so $S'=S_0'$. Therefore, $S'$ is a minimal in-dominating set for $G_1$. \qed
\end{proof}

We have the following corollary, which shows that the in- and out-domination numbers of ILMT tournaments are either constant across time-steps or eventually constant. 

\begin{corollary}\label{cordom}
Let $G_0$ be a tournament and $G_t = \mathrm{ILMT}_{t,s}(G_0)$. We have the following:
\begin{enumerate}
\item $\gamma^+(G_t) = \gamma^+(G_0)$ for all $t\geq 1$.
\item If $G_0$ has no source, then $\gamma^-(G_t) = \gamma^-(G_0)$ for all $t\geq 1$.
\item If $G_0$ has a source and $s$ has infinite support, then $\gamma^-(G_t)=2$ for sufficiently large $t$.
\end{enumerate}
\end{corollary}

\begin{proof}
    For $t\geq 1$, if $S$ is an out-dominating set for $G_{t-1}$, then $S$ is an out-dominating set for $G_{t}$. Each clone has its parent as an out-neighbor and, if $v\in V(G_{t-1})-S$, then $v$ has an out-neighbor in $S$ and so $v'$ does as well. If $S$ is an out-dominating set for $G_{t}$, then the set obtained from $S$ by replacing each clone with its parents is an out-dominating set for $G_{t-1}$. Altogether, we have $\gamma^+(G_{t-1})=\gamma^+(G_{t})$ for all $t\geq 1$. Item (1) follows.

    A similar argument shows that $\gamma^-(G_{t-1})=\gamma^-(G_{t})$ when $s(t)=1$. If $G_0$ has no source, then $G_t$ also has no source for any $t \geq 1$. By Theorem \ref{ILMTindom3}, we have that $\gamma^-(G_{t-1})=\gamma^-(G_{t})$, even when $s(t)= 0$. Item (2) follows.

    For item (3), note that if $u$ is a source in $G_{t-1}$ and $s(t)=1$, then $u'$ is a source in $G_t$. For some $t_0 \geq 0$, suppose that $s(t_0+1) = 0$ and $G_{t_0}$ has a source $u$. We then have that $G_{t_0+1}$ has no source because every parent has its clone as an in-neighbor, every node other than $u$ and $u'$ has $u$ as its in-neighbor, and $u'$ is a sink among the clones. Since $\{u,u'\}$ is in-dominating for $G_{t_0+1}$, we then have that $\gamma^-(G_{t_0+1}) = 2$. Applying item (2) to the sequence $s(t+ t_0+1)$ and base graph $G_{t_0 + 1}$, we have that $\gamma^-(G_t)=\gamma^-(G_{t_0+1})=2$ for all $t>t_0$.\qed
\end{proof}

We next turn to pursuit-evasion games. \emph{Cops and Robbers} is a game played on a tournament $G.$ There are two
players, consisting of a set of cops and a single robber. The game is played over a sequence
of discrete time-steps or rounds indexed by nonnegative integers, with the cops going first
in round 0. The cops and robber occupy nodes. When
a player is ready to move in a round, they must move along an arc to a neighboring node. Players can
pass or remain on their own node. Any subset of cops may move in a given round.
The cops win the game if, after a finite number of rounds, one of them can occupy the
same node as the robber. This situation is called a \emph{capture}. The robber wins if they can
evade capture indefinitely. 

The minimum number of cops required to win is a well-defined
positive integer, called the \emph{cop number} of $G,$ written $c(G).$ For additional background on the
cop number of a graph, see the book \cite{cops}. 

The cop number in ILMT tournaments does not change with a 1-step, as the following theorem shows (with proof omitted). 

\begin{theorem}
If $G_0$ is a tournament and $G_1$ results from a 1-step, then $c(G_0) = c(G_1).$
\end{theorem}

As the domination number upper bounds the cop number, Corollary~\ref{cordom} gives an upper bound of 2 to the cop number of ILMT tournaments whose generating sequence has infinite support. We show that, even after a single 0-step, the cop number is at most 3.

\begin{theorem}\label{0cop}
If $G_0$ is a tournament and $G_1$ results from a 0-step, then we have that
$c(G_1) \leq 3.$
\end{theorem}

\begin{proof}
Let $u$ and $u'$ be a parent node and its clone in $G_1$. 
Suppose that we begin by placing the two cops on $u'$ and the third cop on $u$. 
If the robber begins on a clone node $v'$, then either there is an arc $(u,v)$ in $G_0$, meaning there is an arc $(u,v')$ that the cop on $u$ can use to capture the robber, or there is an arc $(v,u)$ in $G_0$, meaning there is an arc $(u',v')$ that the cop on $u'$ can use to capture the robber.

We may then suppose that the robber begins on a parent node, say $v$. 
We then have that either there is an arc $(u,v)$ in $G_0$, meaning the cop on $u$ can move along this arc to capture the robber, or there is an arc $(v,u)$ in $G_0$, meaning there is an arc $(u',v')$ that one of the two cops on $u'$ can use to move to the node $v'$. 

Since this is the only case in which the robber is not immediately captured, we assume this last case occurs. 
Now, if the robber does not move, the cop on $v'$ can move along the arc $(v',v)$ to capture the robber in the next round. If the robber moves along an arc $(v,w)$, then $w$ must be a parent node to avoid capture by one of the robbers remaining on $u$ and $u'$. However, there is an arc $(v',w)$, and thus, the cop on $v'$ can move along this arc to capture the robber. \qed 
\end{proof}

The following elementary lemma is used to show that the cop number after a 0-step is always at least two. 

\begin{lemma}    [\cite{Halliday_thesis}] If $G$ is a tournament with $c(G)=1$, then $G$ contains a source.  
\end{lemma}

As shown in the proof of Corollary \ref{cordom}, an ILMT tournament resulting from a 0-step does not contain a source. 

\begin{theorem}
    If $G_0$ is a tournament of order at least two and $G_1$ results from a 0-step, then $c(G_1) \geq 2$. 
\end{theorem}

We finish the section by considering colorings of ILMT tournaments. A $k$-\emph{coloring} of a tournament $G$ is a function $\varphi:V(G) \to \{1,2,\dots,k\}$ such that the subtournament induced by the set $\{v\in V(G) : \varphi(v)=i\}$ is a linear order, for $1\leq i \leq k$. The \emph{chromatic number}, written $\chi(G)$, of a tournament $G$ is the minimum $k$ such that $G$ admits a $k$-coloring. 

For a 1-step, the chromatic number remains unchanged.

\begin{theorem}
 If $G_0$ is a tournament and $G_1$ results from a 1-step, then $\chi(G_1) = \chi(G_0)$. 
\end{theorem}

\begin{proof} Let $v_1,v_2,\dots,v_j$ be a color class in a $\chi(G_0)$-coloring of $G_0$, labeled so that $v_{\ell+1},\dots,v_j$ are out-neighbors of $v_{\ell}$, for $1\leq \ell <j$. Let $v'_1,v'_2,\dots,v'_j$ be the clones of $v_1,v_2,\dots,v_j$, respectively. We then have that $v'_1,v_1,v'_2,v_2,\dots,v'_j,v_j$ forms a linear order in $G_1$. 

Repeating this for each color class, we obtain a $\chi(G_0)$-coloring of $G_1$. Hence, $\chi(G_1) \leq \chi(G_0)$. Since $G_0$ is a subtournament of $G_1$, we also have that $\chi(G_1) \geq \chi(G_0)$. \qed
\end{proof}

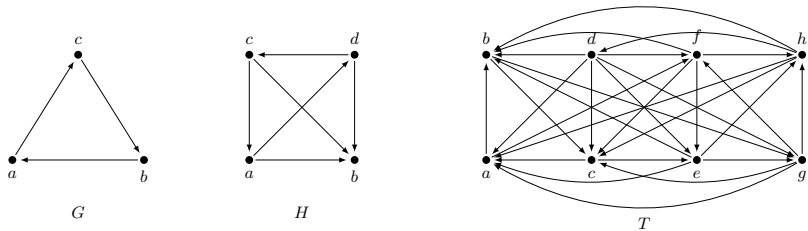
\begin{figure}
    \centering
    \scalebox{0.7}{
    \begin{tikzpicture}[
    vertex/.style={circle, fill=black, inner sep=1.5pt},
    edge/.style={->, >={latex[length=2.5mm, width=1.5mm]}, shorten >=2pt, shorten <=1pt}
    ]

    \begin{scope}
        \node[vertex, label=below:$a$] (a) at (0,0) {};
        \node[vertex, label=below:$b$] (b) at (2.5,0) {};
        \node[vertex, label=above:$c$] (c) at (1.25,2) {};

        \node at (1.25,-1) {$G$};

        \path[edge] (a) edge (c);
        \path[edge] (b) edge (a);
        \path[edge] (c) edge (b);
    \end{scope}

    \begin{scope}[xshift=4.5cm]
        \node[vertex, label=below:$a$] (a) at (0,0) {};
        \node[vertex, label=below:$b$] (b) at (2,0) {};
        \node[vertex, label=above:$c$] (c) at (0,2) {};
        \node[vertex, label=above:$d$] (d) at (2,2) {};

        \node at (1,-1) {$H$};

        \path[edge] (c) edge (a);
        \path[edge] (a) edge (b);
        \path[edge] (d) edge (b);
        \path[edge] (d) edge (c);
        \path[edge] (c) edge (b);
        \path[edge] (a) edge (d);
    \end{scope}

    \begin{scope}[xshift = 9cm]
        \node[vertex, label=below:$a$] (a) at (0,0) {};
        \node[vertex, label=above:$b$] (b) at (0,2) {};
        \node[vertex, label=below:$c$] (c) at (2,0) {};
        \node[vertex, label=above:$d$] (d) at (2,2) {};
        \node[vertex, label=below:$e$] (e) at (4,0) {};
        \node[vertex, label=above:$f$] (f) at (4,2) {};
        \node[vertex, label=below:$g$] (g) at (6,0) {};
        \node[vertex, label=above:$h$] (h) at (6,2) {};

        \node at (3,-1.2) {$T$};

        \path[edge] (a) edge (b);
        \path[edge] (c) edge (a);
        \path[edge] (d) edge (a);
        \path[edge] (e) edge [bend left=20] (a);
        \path[edge] (a) edge (f);
        \path[edge] (g) edge [bend left] (a);
        \path[edge] (h) edge (a);
        \path[edge] (b) edge (c);
        \path[edge] (d) edge (b);
        \path[edge] (e) edge (b);
        \path[edge] (f) edge [bend right=20] (b);
        \path[edge] (b) edge (g);
        \path[edge] (h) edge [bend right] (b);
        \path[edge] (d) edge (c);
        \path[edge] (c) edge (e);
        \path[edge] (f) edge (c);
        \path[edge] (g) edge [bend left=20] (c);
        \path[edge] (h) edge (c);
        \path[edge] (d) edge (f);
        \path[edge] (d) edge (e);
        \path[edge] (d) edge (g);
        \path[edge] (h) edge [bend right=20] (d);
        \path[edge] (f) edge (e);
        \path[edge] (e) edge (g);
        \path[edge] (e) edge (h);
        \path[edge] (f) edge (h);
        \path[edge] (g) edge (f);
        \path[edge] (g) edge (h);
    \end{scope}

    \end{tikzpicture}
    }
    \caption{Tournaments $G$, $H$, and $T$, each with chromatic number 2, for which a 0-step increases the chromatic number by 0, 1, and 2, respectively.}
    \label{fig:placeholder}
\end{figure}

For a 0-step, we found examples in which the chromatic number remains the same, increases by 1, or increases by 2. See Figure~2. In general, the chromatic number can at most double as a result of a 0-step. For $t\geq 1$ and $2t$-many 0-steps, we have the following upper bound on the chromatic number of $G_{2t}$, obtained using a partition of $\mathbb{Z}^{2t}_2$ into sets whose dot products between elements are either always 0 or always 1 (mod 2). By labeling each node of $G_{2t}$ with a binary string of length $2t$ and its parent node in $G_0$, we can describe the color classes in $G_{2t}$ using color classes in an optimal coloring of the base tournament and the partition of $\mathbb{Z}^{2t}_2$ described above. The full proof is omitted for space considerations. 

\begin{theorem}
    Let $G_0$ be a tournament and, for $t\geq 0$, let $G_{2t}$ be the tournament obtained by performing $2t$-many 0-steps. We then have that
    $$\chi(G_{2t}) \leq (2^{t+1}-1)\chi(G_0).$$ 
\end{theorem}

We define a sequence $S_i$ of tournaments as follows, first given in \cite{hero_coloring}. Let $S_1$ be the one-node tournament. For $i\geq 2$, let $S_i = \Delta(S_{i-1},S_{i-1},1)$ be the tournament formed by taking the disjoint union of two copies of $S_{i-1}$, call them $S^{(1)}_{i-1}$ and $S^{(2)}_{i-1}$, and a single node $v$. Add all possible arcs from $S^{(1)}_{i-1}$ to $S^{(2)}_{i-1}$, from $S^{(2)}_{i-1}$ to $v$, and from $v$ to $S^{(1)}_{i-1}$. Note that $S_i$ contains $2^i-1$ nodes. 

\begin{theorem} [\cite{hero_coloring}] For $i\geq 1$, $\chi(S_i)\geq i$. 
    \label{thm:S_i}
\end{theorem}

We have the following corollary, which shows that the chromatic numbers of ILMT tournaments are unbounded. 

\begin{corollary}
   Let $i\geq 1$ and $G_0$ be a tournament with $n_0 \geq 2^i-1$ nodes. Let $G_t$ be the tournament obtained by performing $t$-many 0-steps. We then have that $\log_2(t+1) \leq \chi(G_t).$   
\end{corollary}

\begin{proof}
By Theorem~\ref{lemma:0steps}, at most $t=2^i-1$ many 0-steps are required to guarantee that $G_t$ contains the tournament $S_i$ as a subtournament. This subtournament requires $i=\log_2(t+1)$ colors by Theorem \ref{thm:S_i}, and therefore $G_t$ does as well. \qed
\end{proof}

\section{Discussion and Future Work}

The ILMT model produces tournaments with small diameters and increasing connectivity. Given a generating sequence with infinite support, the resulting ILMT tournaments contain all motifs as subtournaments. Further, such generating sequences give rise to continuum-many non-isomorphic quasirandom families of tournaments. We also determined how domination numbers, the cop number, and the chromatic number behave under the 0- and 1-steps in the model.

Several directions remain open. A probabilistic version of ILMT, where each time-step applies a $0$- or $1$-step with given probability, would be interesting to analyze. The precise growth of the chromatic number under general infinite sequences also remains open. While we have a logarithmic lower bound for the chromatic number of ILMT tournaments, their actual chromatic number is likely much larger. It is unclear how close the upper bound given in Theorem 12 is to being optimal. 

While ILMT applies to dense communities in complex networks, a natural extension of ILMT is to sparser oriented graphs. For example, in a 1-step, include the arcs $(x',u)$ for each node $x$ and out-neighbor $u$ of $x$ in $G_{t-1}$; in a 0-step, include the arcs $(x',v)$ for each node $x$ and in-neighbor $v$ of $x$ in $G_{t-1}$. See Figure~\ref{ILMO}.
\begin{figure}[H]
    \centering
    \scalebox{0.7}{
\begin{tikzpicture}[
    vtx/.style={circle, fill=black, inner sep=1.5pt},
    arr/.style={->, >={latex[length=2.5mm, width=1.5mm]}, shorten >=1pt, shorten <=1pt}
]
    \begin{scope}
        \node[vtx, label=below:{$a$}] (g1_a1) at (0, 0) {};
        \node[vtx, label=above:{$b$}] (g1_a) at (0, 2) {};

        \node at (0, -1) {$G_0$}{};
        
        \draw[arr] (g1_a1) to (g1_a);
    \end{scope}

    \begin{scope}[xshift=2.5cm]
        \node[vtx, label=below:{$a$}] (g2_a) at (0, 0) {};
        \node[vtx, label=above:{$b$}] (g2_b) at (0, 2) {};
        \node[vtx, label=below:{$a'$}] (g2_a1) at (2, 0) {};
        \node[vtx, label=above:{$b'$}] (g2_b1) at (2, 2) {};

        \node at (1, -1) {$G_1$}{};
        
        \draw[arr] (g2_a1) to (g2_a);    
        \draw[arr] (g2_b1) to (g2_b);    
        \draw[arr] (g2_a1) to (g2_b);    
        \draw[arr] (g2_a) to (g2_b);     
    \end{scope}

    \begin{scope}[xshift=7cm]
        \node[vtx, label=below:{$a$}] (a) at (0, 0) {};
        \node[vtx, label=above:{$b$}] (b) at (0, 2) {};
        \node[vtx, label=below:{$a'$}] (a1) at (2, 0) {};
        \node[vtx, label=above:{$b'$}] (b1) at (2, 2) {};
        \node[vtx, label=below:{$a''$}] (a2) at (4, 0) {};
        \node[vtx, label=above:{$b''$}] (b2) at (4, 2) {};
        \node[vtx, label=below:{$(a')'$}] (a3) at (6, 0) {};
        \node[vtx, label=above:{$(b')'$}] (b3) at (6, 2) {};
    
        \node at (3, -1) {$G_2$};
    
        \draw[arr] (a) to (b);
        \draw[arr] (a1) to (a);
        \draw[arr] (a1) to (b);
        \draw[arr] (b1) to (b);
        \draw[arr] (a2) to (a1);
        \draw[arr] (a2) to [bend left=35] (a);
        \draw[arr] (b2) to (a1);
        \draw[arr] (b2) to (a);
        \draw[arr] (b2) to [bend right=35] (b);
        \draw[arr] (b2) to (b1);
        \draw[arr] (a3) to [bend left=35] (a1);
        \draw[arr] (b3) to [bend right=35] (b1);
        
    \end{scope}

\end{tikzpicture}}

    \caption{A sequence of oriented graphs formed by a 1-step and then a 0-step.}
    \label{ILMO}
\end{figure}
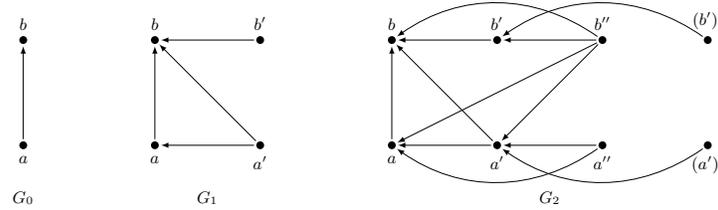
Oriented graphs generated by this model will not be tournaments, but will provably generate densifying families. 
It would be interesting to analyze the complex network and graph-theoretic properties of such iterated oriented models in future work.


\begin{thebibliography}{99}

\bibitem{bang-jensen-gutin-2009} J.\ Bang-Jensen, G.\ Gutin, \emph{Digraphs}, Springer-Verlag London, Ltd., London, 2009.

\bibitem{ba} A.\ Barab\'{a}si, R.\ Albert, Emergence of scaling in random networks, \emph{Science} \textbf{286} (1999) 509--512.

\bibitem{ILH} N.\ Behague, A.\ Bonato, M.\ Huggan, R.\ Malik, T.\ Marbach, The iterated local transitivity model for hypergraphs, \textit{Discrete Applied Mathematics} \textbf{337} (2023) 106--119.

\bibitem{hero_coloring} E.\ Berger, K.\ Choromanski, M.\ Chudnovsky, J.\ Fox, M.\ Loebl, A.\ Scott, P.\ Seymour, S.\ Thomass\'e,
Tournaments and coloring,
\textit{Journal of Combinatorial Theory, Series B} \textbf{103} (2013) 1--20.

\bibitem{bol} B.\ Bollobás, O.\ Riordan, J.\ Spencer, G.\ Tusnády, The degree sequence of a scale-free random graph process,
\textit{Random Structures \& Algorithms} \textbf{18} (2001) 279--290.

\bibitem{bbook} A.\ Bonato, \emph{A Course on the Web Graph}, American Mathematical Society, Providence, Rhode Island, 2008.

\bibitem{ILTT} A.\ Bonato, K.\ Chaudhary, The iterated local transitivity model for tournaments,
In: \textit{Proceedings of WAW'23}, 2023. 

\bibitem{ILM} A.\ Bonato, H.\ Chuangpishit, S.\ English, B.\ Kay, E.\ Meger, The iterated local model for social networks, \textit{Discrete Applied Mathematics} \textbf{284} (2020) 556--571.

\bibitem{ILDT} A.\ Bonato, D.\ W.\ Cranston, M.\ A.\ Huggan, T.\ Marbach, R.\ Mutharasan, The iterated local directed transitivity model for social networks, In: \textit{Proceedings of WAW'20}, 2020.

\bibitem{FM} A.\ Bonato, R.\ Cushman, T.\ Marbach, Z.\ Zhang, The frustum network model based on clique extension, \textit{Journal of Combinatorial Optimization} \textbf{47} (2024).

\bibitem{ILT} A.\ Bonato, N.\ Hadi, P.\ Pralat, C.\ Wang,
Models of on-line social networks, \textit{Internet Mathematics} \textbf{6} (2011) 285--313.

\bibitem{ILAT} A.\ Bonato, E.\ Infeld, H.\ Pokhrel, P.\ Pralat, Common adversaries form alliances: modelling complex networks via anti-transitivity,
In: \textit{Proceedings of WAW'17}, 2017.

\bibitem{IGM} A.\ Bonato, E.\ Meger,
Iterated global models for complex networks,
In: \textit{International Workshop on Algorithms and Models for the Web-Graph} (2020) 135--144.

\bibitem{cops} A.\ Bonato, R.J.\ Nowakowski, \emph{The Game of Cops and Robbers on Graphs}, American Mathematical
Society, Providence, Rhode Island, 2011.

\bibitem{at} A.\ Bonato, A.\ Tian, Complex networks and social networks, invited book chapter in: \emph{Social Networks},  editor E.\ Kranakis, 2011.

\bibitem{mscthesis} K.\ Chaudhary,
\emph{The Iterated Local Transitivity Model For Tournaments}, Master's thesis, 2023.

\bibitem{fan}
F.R.K.\ Chung, R.L.\ Graham,
Quasi-random tournaments,
\textit{Journal of Graph Theory} \textbf{15} (1991) 173–198.

\bibitem{easl} D.\ Easley, J.\ Kleinberg,
\emph{Networks, crowds, and markets: Reasoning about a highly connected world},
Cambridge University Press, 2010.

\bibitem{guha} R.\ Guha, R.\ Kumar, P.\ Raghavan, A.\ Tomkins,
Propagation of trust and distrust,
In: \textit{Proceedings of the 13th international conference on World Wide Web}, 2004.

\bibitem{Halliday_thesis} F.\ Halliday,
\emph{Cops, robbers and pre-calculus skills},
Masters Thesis, 2019.

\bibitem{raz}
R.\ Hancock, A.\ Kabela, D.\ Král', T.\ Martins, R.\ Parente, F.\ Skerman, J.\ Volec,
No additional tournaments are quasirandom-forcing.
\textit{European Journal of Combinatorics} \textbf{108} (2023) 103632.

\bibitem{he} F.\ Heider, \emph{The Psychology of Interpersonal Relations}, John Wiley \& Sons, 1958.

\bibitem{lesk} J.\ Leskovec, D.\ Huttenlocher, J.\ Kleinberg,
Signed networks in social media,
In: \textit{Proceedings of the SIGCHI conference on human factors in computing systems}, 2010.

\bibitem{itind} E.\ Meger, A.\ Raz,
The Iterative Independent Model,
\textit{Discrete Applied Mathematics} \textbf{341} (2023) 242--256.

\bibitem{Privault2018}N. Privault, Long-Run Behavior of Markov Chains. {\em Understanding Markov Chains: Examples And Applications}. pp. 163-188 (2018).

\bibitem{Milo}
R.\ Milo, S.\ Shen-Orr, S.\ Itzkovitz, N.\ Kashtan, D.\ Chklovskii, U.\ Alon,
Network motifs: simple building blocks of complex networks.
\textit{Science} \textbf{298} (2002) 824--827.

\bibitem{scott} J.P.\ Scott, \emph{Social Network Analysis: A Handbook}, Sage Publications 
Ltd, London, 2000.

\bibitem{west} D.B.\ West, \emph{Introduction to Graph Theory, 2nd edition}, Prentice Hall, 2001.

\end{thebibliography}
\end{document}